\documentclass{llncs}
\usepackage{amssymb}
\usepackage{amsmath}
\usepackage{mathrsfs}
\usepackage{graphicx}
\usepackage{subfigure}
\usepackage{tikz}
\usepackage{xspace}
\usepackage{enumerate}
\usetikzlibrary{automata,positioning}
\tikzstyle{every path}=[arrows=-latex]
\tikzstyle{every loop}=[arrows=-latex]
\tikzstyle{accepting}=[double distance=2pt]

\newcommand{\SO}{\textsc{Shift-Or}\xspace}
\newcommand{\SA}{\textsc{Shift-And}\xspace}
\newcommand{\BNDM}{\textsc{BNDM}\xspace}
\newcommand{\bpstyle}[1]{\mathsf{#1}}
\newcounter{instr}
\newcommand{\ninstr}{\refstepcounter{instr}\theinstr.}

\newcommand{\AND}{\mathrel{\&}}
\newcommand{\OR}{\mathrel{|}}
\newcommand{\NOT}{\mathop{\sim}}

\newcommand{\defAs}{\:=\!\!_{\mbox{\tiny Def}}\:}
\newcommand{\ceil}[1]{\lceil #1 \rceil}

\newcommand{\Suff}{\mathit{Suff}\xspace}

\begin{document}

\author{Emanuele Giaquinta\inst{1}}
\institute{Department of Computer Science and Engineering, Aalto University, Finland
	  \email{emanuele.giaquinta@aalto.fi}}
\title{Run-Length Encoded Nondeterministic KMP and Suffix Automata}

\maketitle

\begin{abstract}
We present a novel bit-parallel representation, based on the
run-length encoding, of the nondeterministic KMP and suffix automata
for a string $P$ with at least two distinct symbols. Our method is
targeted to the case of long strings over small alphabets and
complements the method of Cantone et al. (2012), which is effective
for long strings over large alphabets. Our encoding requires
$O((\sigma + m)\lceil \rho / w\rceil)$ space and allows one to
simulate the automata on a string in time $O(\lceil \rho / w\rceil)$
per transition, where $\sigma$ is the alphabet size, $m$ is the length
of $P$, $\rho$ is the length of the run-length encoding of $P$ and $w$
is the machine word size in bits. The input string can be given in
either unencoded or run-length encoded form.
\end{abstract}

\section{Introduction}

The string matching problem consists in finding all the occurrences of
a string $P$ of length $m$ in a string $T$ of length $n$, both over a
finite alphabet $\Sigma$ of size $\sigma$. The matching can be either
exact or approximate, according to some metric which measures the
closeness of a match. The finite automata for the language $\Sigma^*
P$ (prefix automaton) and $\Suff(P)$ (suffix automaton), where
$\Suff(P)$ is the set of suffixes of $P$, are the main building blocks
of very efficient algorithms for the exact and approximate string
matching problem. Two fundamental algorithms for the exact problem, based on the deterministic version of these automata,
are the KMP and BDM algorithms, which run in $O(n)$ and $O(nm)$
worst-case time, respectively, using $O(m)$
space~\cite{DBLP:journals/siamcomp/KnuthMP77,DBLP:books/ox/CrochemoreR94}.
In the average case, the BDM algorithm achieves the optimal
$\mathcal{O}(n\log_{\sigma}(m)/m)$ time bound. The nondeterministic
version of the prefix and suffix automata can be simulated using an
encoding, known as bit-parallelism, based on bit-vectors and
word-level parallelism~\cite{DBLP:journals/cacm/Baeza-YatesG92}. The
variants of the KMP algorithm based on the nondeterministic prefix
automaton, known as \SO and \SA, run in $O(n\ceil{m/w})$ worst-case
time and use $O(\sigma\ceil{m/w})$ space, where $w$ is the machine
word size in
bits~\cite{DBLP:journals/cacm/Baeza-YatesG92,DBLP:journals/cacm/WuM92}.
Similarly, the variant of the BDM algorithm based on the
nondeterministic suffix automaton, known as \BNDM, runs in
$\mathcal{O}(nm\ceil{m/w})$ worst-case time and uses
$O(\sigma\ceil{m/w})$ space~\cite{DBLP:journals/jea/NavarroR00}. In
the average case, the \BNDM algorithm runs in
$\mathcal{O}(n\log_{\sigma}(m)/w)$ time, which is suboptimal for
patterns whose length is greater than $w$. There also exists a variant
of \SO which achieves $\mathcal{O}(n\log_{\sigma}(m)/w)$ time in the
average case~\cite{DBLP:journals/jda/FredrikssonG09}. As for the
approximate string matching problem, there are also various algorithms
based on the nondeterministic prefix and suffix
automata~\cite{DBLP:journals/cacm/WuM92,DBLP:journals/jea/NavarroR00,DBLP:journals/algorithmica/Baeza-YatesN99,DBLP:journals/ipl/Hyyro08,DBLP:journals/algorithmica/HyyroN05}.

In general, the bit-parallel algorithms are suboptimal if compared to
the their ``deterministic'' counterparts in the case $m > w$, because
of the $\ceil{m/w}$ additional term in the time complexity. A way to
overcome this problem is to use a filtering method, namely, searching
for the prefix of $P$ of length $w$ and verifying each occurrence with
the naive algorithm. Assuming uniformly random strings, the average
time complexity of \SA and \BNDM with this method is $O(n)$ and $O(n
\log_{\sigma} w / w)$, respectively. Recently, a few approaches were
proposed to improve the case of long patterns. In 2010 Durian et al.
presented three variants of \BNDM tuned for the case of long patterns,
two of which are optimal in the average
case~\cite{DBLP:conf/wea/DurianPST10}. In the same year, Cantone et
al. presented a different encoding of the prefix and suffix automata,
based on word-level parallelism and on a particular factorization on
strings~\cite{DBLP:journals/iandc/CantoneFG12}. The general approach
is to devise, given a factorization $f$ on strings, a bit-parallel
encoding of the automata based on $f$ such that one transition can be
performed in $O(\ceil{|f(P)| / w})$ time instead of $O(\ceil{m/ w})$,
at the price of more space. The gain is two-fold: i) if $|f(P)| <
|P|$, then the overhead of the simulation is reduced. In particular,
there is no overhead if $|f(P)| \le w$, which is preferable if $|f(P)|
< |P|$; ii) if we use the filtering method, we can search for the
longest substring $P'$ of $P$ such that $|f(P')|\le w$. This yields
$O(n \log_{\sigma} |P'| / |P'|)$ average time for \BNDM, which is
preferable if $|P'| > w$. The method of Cantone et al. is effective
for long patterns over large alphabets. Their factorization is such
that $\ceil{m/\sigma}\le |f(P)|\le m$ and their encoding requires
$O(\sigma^2 \ceil{|f(P)| / w})$ space.

In this paper we present a novel encoding of the prefix and suffix
automata, based on this approach, where $f(P)$ is the run-length
encoding of $P$, provided that $P$ has at least two distinct symbols.
The run-length encoding of a string is a simple encoding where each
maximal consecutive sequence of the same symbol is encoded as a pair
consisting of the symbol plus the length of the sequence. Our encoding
requires $O((\sigma + m)\ceil{\rho / w})$ space and allows one to
simulate the automata in $O(\ceil{\rho / w})$ time per transition,
where $\rho$ is the length of the run-length encoding of $P$. While
the present algorithm uses the run-length encoding, the input string
can be given in either unencoded or run-length encoded form. The
run-length encoding is suitable for strings over small alphabets.
Therefore, our method complements the one of Cantone et al., which is
effective for large alphabets.

\section{Notions and Basic Definitions}

Let $\Sigma$ be a finite alphabet of symbols and let $\Sigma^*$ be the
set of strings over $\Sigma$. The empty string $\varepsilon$ is a
string of length $0$. Given a string $S$, we denote with $|S|$ the
length of $S$ and with $S[i]$ the $i$-th symbol of $S$, for $0\le i <
|S|$. The concatenation of two strings $S$ and $\bar{S}$ is denoted by
$S\bar{S}$. Given two strings $S$ and $\bar{S}$, $S$ is a substring of
$\bar{S}$ if there are indices $0\le i,j < |S|$ such that $\bar{S} =
S[i] ... S[j]$. If $i = 0$ ($j = |S| - 1$) then $\bar{S}$ is a prefix
(suffix) of $S$. The set $\Suff(S)$ is the set of all suffixes of $S$.
We denote by $S[i .. j]$ the substring $S[i] .. S[j]$ of $S$. For $i >
j$ $S[i .. j] = \varepsilon$. We denote by $S^k$ the concatenation of
$k$ $S$'s, for $S\in\Sigma^*$ and $k\ge 1$. The string $S^r$ is the
reverse of the string $S$, i.e., $S^r=S[|S|-1]S[|S|-2]\ldots S[0]$.

Given a string $P\in \Sigma^*$ of length $m$, we denote by
$\mathcal{A}(P)=(Q,\Sigma,\delta,q_0,F)$ the nondeterministic finite
automaton (NFA) for the language $\Sigma^*P$ of all strings in $\Sigma^{*}$
whose suffix of length $m$ is $P$, where:
\begin{itemize}
    \item
    $Q=\{q_0,q_1,\ldots, q_m\}$ \qquad ($q_{0}$ is the initial state)

    \item the transition function $\delta: Q \times \Sigma
    \longrightarrow \mathscr{P}(Q)$ is defined by:
    $$
    \delta(q_{i},c) \defAs \begin{cases}
    \{q_{0},q_{1}\} & \text{if } i = 0 \text{ and } c = P[0]\\
    \{q_{0}\} & \text{if } i = 0 \text{ and } c \neq P[0]\\
    \{q_{i+1}\} & \text{if } 1 \le i < m \text{ and } c = P[i]\\
    \emptyset & \text{otherwise}
    \end{cases}
    $$
    \item
    $F=\{q_m\}$ \qquad ($F$ is the set of final states).
\end{itemize}

Similarly, we denote by $\mathcal{S}(P)=(Q,\Sigma,\delta,I,F)$ the
nondeterministic suffix automaton with $\varepsilon$-transitions for
the language $\Suff(P)$ of the suffixes of $P$, where:
\begin{itemize}
    \item
    $Q=\{I,q_0,q_1,\ldots, q_m\}$ \qquad ($I$ is the initial state)

    \item the transition function $\delta: Q \times (\Sigma \cup
    \{\varepsilon\}) \longrightarrow \mathscr{P}(Q)$ is defined by:
    $$
    \delta(q,c) \defAs \begin{cases}
    \{q_{i+1}\} & \text{if } q = q_{i} \text{ and } c = P[i] ~~(0
    \le i < m)\\
    Q & \text{if } q = I  \text{ and } c =
    \varepsilon\\
    \emptyset & \text{otherwise}
    \end{cases}
    $$
    \item
    $F=\{q_m\}$ \qquad ($F$ is the set of final states).
\end{itemize}

We use the notation $q_I$ to indicate the initial state of the
automaton, i.e., $q_I$ is $q_0$ for $\mathcal{A}(P)$ and $I$ for
$\mathcal{S}(P)$. The valid configurations $\delta^*(q_{I},S)$ which
are reachable by the automata $\mathcal{A}(P)$ and $\mathcal{S}(P)$ on
input $S \in \Sigma^{*}$ are defined recursively as follows:
$$
    \delta^*(q_{I},S)=_{\mathit{Def}}\begin{cases}
    E(q_I) & \text{if $S=\varepsilon$,}\\
    \bigcup_{q' \in \delta^{*}(q_{I},S')}\delta(q',c) & \text{if
    $S=S'c$,
    for some $c\in \Sigma$ and $S' \in \Sigma^{*}$.}
    \end{cases}
$$
where $E(q_I)$ denotes the $\varepsilon$-closure of $q_I$. Given a
string $P$, a run of $P$ is a maximal substring of $P$ containing
exactly one distinct symbol. The run-length encoding (RLE) of a string
$P$, denoted by $\textsc{rle}(P)$, is a sequence of pairs (runs)
$\langle (c_0, l_0), (c_2, l_2,), \ldots, (c_{\rho-1}, l_{\rho-1})
\rangle$ such that $c_i\in \Sigma$, $l_i\ge 1$, $c_i\neq c_{i+1}$ for
$0 \le i < \rho$, and $P = c_0^{l_0} c_1^{l_1} \ldots
c_{\rho-1}^{l_{\rho-1}}$. The starting (ending) position in $P$ and
length of the run $(c_i, l_i)$ are $\alpha_P(i) = \sum_{j=0}^{i-1}
l_j$ ($\beta_P(i) = \sum_{j=0}^{i} l_j - 1$) and $\ell_P(i) = l_i$, for
$i=0,\ldots,\rho-1$. We also put $\alpha_P(\rho) = |P|$.

Finally, we recall the notation of some bitwise infix operators on
computer words, namely the bitwise \texttt{and} ``$\&$'', the bitwise
\texttt{or} ``$|$'', the \texttt{left shift} ``$\ll$'' operator (which
shifts to the left its first argument by a number of bits equal to its
second argument), and the unary bitwise \texttt{not} operator
``$\NOT$''.

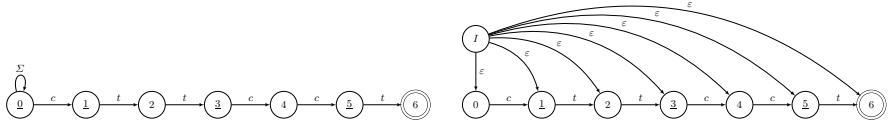
\begin{figure}[t]
\begin{center}
\subfigure{
\resizebox{0.47\textwidth}{!}{
\begin{tikzpicture}[node distance=2cm,auto]
\node[state] (q_0) {$\underline{0}$};
\node[state] (q_1) [right of=q_0] {$\underline{1}$};
\node[state] (q_2) [right of=q_1] {$2$};
\node[state] (q_3) [right of=q_2] {$\underline{3}$};
\node[state] (q_4) [right of=q_3] {$4$};
\node[state] (q_5) [right of=q_4] {$\underline{5}$};
\node[state,accepting] (q_6) [right of=q_5] {$6$};

\path[->] (q_0) edge node {$c$} (q_1)
                edge [loop above] node {$\Sigma$} ()
          (q_1) edge node {$t$} (q_2)
          (q_2) edge node {$t$} (q_3)
          (q_3) edge node {$c$} (q_4)
          (q_4) edge node {$c$} (q_5)
          (q_5) edge node {$t$} (q_6);

\end{tikzpicture}
}
}
\subfigure{
\resizebox{0.47\textwidth}{!}{
\begin{tikzpicture}[node distance=2cm,auto]
\node[state] (q_0) {$0$};
\node[state] (I) [above of=q_0] {$I$};
\node[state] (q_1) [right of=q_0] {$\underline{1}$};
\node[state] (q_2) [right of=q_1] {$2$};
\node[state] (q_3) [right of=q_2] {$\underline{3}$};
\node[state] (q_4) [right of=q_3] {$4$};
\node[state] (q_5) [right of=q_4] {$\underline{5}$};
\node[state,accepting] (q_6) [right of=q_5] {$6$};

\path[->] (q_0) edge node {$c$} (q_1)
          (q_1) edge node {$t$} (q_2)
          (q_2) edge node {$t$} (q_3)
          (q_3) edge node {$c$} (q_4)
          (q_4) edge node {$c$} (q_5)
          (q_5) edge node {$t$} (q_6)
          (I) edge node {$\varepsilon$} (q_0)
          (I) edge [bend left] node {$\varepsilon$} (q_1)
          (I) edge [bend left] node {$\varepsilon$} (q_2)
          (I) edge [bend left] node {$\varepsilon$} (q_3)
          (I) edge [bend left] node {$\varepsilon$} (q_4)
          (I) edge [bend left] node {$\varepsilon$} (q_5)
          (I) edge [bend left] node {$\varepsilon$} (q_6);
\end{tikzpicture}
}
}
\caption{(a) The automata $\mathcal{A}(P)$ and $\mathcal{S}(P)$ for
  the pattern $P = cttcct$. The state labels corresponding to
  the starting positions of the runs of $\textsc{RLE}(P)$ are
  underlined.}
\label{fig:example}
\end{center}
\end{figure}

\section{The \SA and \BNDM algorithms}

In this section we briefly describe the \SA and \BNDM algorithms.
Given a pattern $P$ of length $m$ and a text $T$ of length $n$, the
\SA and \BNDM algorithms find all the occurrences of $P$ in $T$. The
\SA algorithm works by simulating the $\mathcal{A}(P)$ automaton on
$T$ and reporting all the positions $j$ in $T$ such that the final
state of $\mathcal{A}(P)$ is active in the corresponding configuration
$\delta^*(q_I, T[0\,\ldots\,j])$. Instead, the \BNDM algorithm works
by sliding a window of length $m$ along $T$. For a given window ending
at position $j$, the algorithm simulates the automaton
$\mathcal{S}(P^r)$ on $(T[j-m+1\,\ldots\,j])^r$. Based on the
simulation, the algorithm computes the length $k$ and $k'$ of the
longest suffix of $T[j-m+1\,\ldots\,j]$ which is a prefix and a proper
prefix, respectively, of $P$ (i.e., a suffix of $P^r$). If $k = m$
then $T[j-m+1\,\ldots\,j] = P$ and the algorithm reports an occurrence
of $P$ at position $j$. The window is then shifted by $m-k'$ positions
to the right, so as to align it with the longest proper prefix of $P$
found. The automata are simulated using an encoding based on
bit-vectors and word-level parallelism. The algorithms run in
$O(n\ceil{m/w})$ and $O(nm\ceil{m/w})$ time, respectively, using
$O(\sigma\ceil{m/w})$ space, where $w$ is the word size in bits.

\section{RLE-based encoding of the Nondeterministic KMP and suffix automata}

Given a string $P$ of length $m$ defined over an alphabet $\Sigma$ of
size $\sigma$, let $\textsc{rle}(P) = \langle (c_0, l_0), (c_1, l_1,),
\ldots, (c_{\rho-1}, l_{\rho-1}) \rangle$ be the run-length encoding
of $P$. In the following, we describe how to simulate the
$\mathcal{A}(P)$ and $\mathcal{S}(P)$ automata, using word-level
parallelism, on a string $S$ of length $n$ in $O(\ceil{\rho / w})$
time per transition and $O((m + \sigma)\ceil{\rho / w})$ space. We
recall that the simulation of the automaton $\mathcal{A}(P)$ on a
string $S$ detects all the prefixes of $S$ whose suffix of length $m$
is $P$. Similarly, the simulation of the automaton $\mathcal{S}(P)$
detects all the prefixes of $S$ which are suffixes of $P$.

Let $\mathcal{I}(S) = \{ \alpha_S(i)\ |\ 0\le i\le |\textsc{rle}(S)|
\}$ be the set of starting positions of the runs of $S$, for a given
string $S$. Note that $0\in \mathcal{I}(S)$. Given a string $S$, we
denote with $D_j = \delta^*(q_I, S[0\,\ldots\, j-1])$ the
configuration of $\mathcal{A}(P)$ or $\mathcal{S}(P)$ after reading
$S[0\,\ldots\, j-1]$, for any $0\le j \le |S|$. The main idea of our
algorithm is to compute the configurations $D_j$ corresponding to
positions $j\in\mathcal{I}(S)$ only. We show that, for any two
consecutive positions $j,j'\in\mathcal{I}(T)$, 1) we need only the
states of the automaton corresponding to positions in $\mathcal{I}(P)$
to compute $D_{j'}$ from $D_j$; 2) if the pattern includes at least
distinct two symbols (i.e., if $\rho\ge 2$), there can be at most one
occurrence of $P$ ending in a position between $j$ and $j'$ (or
equivalently spanning a proper prefix of a given run of $S$).
Similarly, there can be at most one prefix of $S$ ending in a position
between $j$ and $j'$ which corresponds to a suffix of $P$ with at
least two distinct symbols. We start with the following Lemma:
\begin{lemma}
  Let $j\in\mathcal{I}(S)$. Then, for any $q_i\in D_j$ such that
  $i\notin\mathcal{I}(P)$, we have $\delta(q_i, S[j]) = \emptyset$.
\end{lemma}
\begin{proof}
  Let $q_i\in D_j$ with $i\notin\mathcal{I}(P)$. By definition of
  $q_i$ it follows that $S[j-1] = P[i-1]$ and $P[i-1] = P[i]$,
  respectively. Moreover, by $j\in\mathcal{I}(S)$, we have $S[j]\neq
  S[j-1]$. Suppose that $\delta(q_i, S[j])\neq \emptyset$, which
  implies $S[j] = P[i]$. Then we have $S[j] = P[i] = P[i-1] = S[j-1]$,
  which yields a contradiction.\qed
\end{proof}
This Lemma states that, for any $j\in\mathcal{I}(S)$, any state
$q_i\in D_j$ with $i\notin\mathcal{I}(P)$ is dead, as no transition is
possible from it on $S[j]$. Figure~\ref{fig:example} shows the
automata $\mathcal{A}(P)$ and $\mathcal{S}(P)$ for $P = cttcct$; the
state labels corresponding to indexes in $\mathcal{I}(P)$ are
underlined.

We assume that $P$ has at least two distinct symbols. The following
Lemma shows that, under this assumption, there can be at most one
configuration $D$ containing the final state $q_m$ with index between
$\alpha_S(i)+1$ and $\beta_S(i)+1$ in $S$, for any $1\le i\le
|\textsc{rle}(S)|$ (note that $i\ge 1$ implies that the corresponding
prefix of $S$ in the language has at least two distinct symbols).
\begin{lemma}\label{lemma:single}
  Let $i\in \{1, \ldots, |\textsc{rle}(S)| - 1 \}$. If $\rho\ge 2$,
  there exists at most one $j'$ in the interval $[\alpha_S(i),
    \beta_S(i)]$ such that $q_m\in D_{j'+1}$.
\end{lemma}
\begin{proof}
  Let $j'\in [j_1, j_2]$ such that $q_m\in D_{j'+1}$, where $j_1 =
  \alpha_S(i)$ and $j_2 = \beta_S(i)$. This corresponds to i) $S[j' -
    |P| + 1\,\ldots\, j'] = P$ for $\mathcal{A}(P)$ and to ii) and
  $S[0\,\ldots\, j']\in \Suff(P)$ for $\mathcal{S}(P)$. Since $\rho\ge
  2$ and $i\ge 1$, in both cases, for this to hold we must have $S[j'
    - k] = c_{\rho-1}$, for $k = 0, \ldots, l_{\rho-1} - 1$, and $S[j'
    - l_{\rho-1}] = c_{\rho-2}$, where $c_{\rho-2}\neq c_{\rho-1}$. By
  definition of $j_1$, $S[j_1] = S[j_1+1] = \ldots = S[j']$ and
  $S[j_1-1]\neq S[j_1]$. Hence, the only possibility is $j' = j_1 +
  l_{\rho-1} - 1$.
\end{proof}
Specifically, the only configuration containing $q_m$, if any,
corresponds to index $\alpha_S(i) + l_{\rho-1}$. By definition of
$D_j$ and by Lemma $1$, we have
\begin{equation}\label{eq:delta}
\begin{array}{ll}
D_{\alpha_S(j+1)} & = \delta^*(q_I, S[0\,\ldots\,\beta_S(j)]) \\
               & = \bigcup_{q\in D_{\alpha_S(j)}}\delta^*(q, S[\alpha_S(j)\,\ldots\,\beta_S(j)]) \\
               & = \bigcup_{q\in D_{\alpha_S(j)}\cap\{q_i\ |\ i\in\mathcal{I}(P)\}}\delta^*(q, S[\alpha_S(j)]^{\ell_S(j)}) \\
\end{array}
\end{equation}
for any position $\alpha_S(j+1)$. Moreover, by Lemma $2$, there can be
at most one configuration $D$ with index between $\alpha_S(j)+1$ and
$\beta_S(j)+1$ containing the final state $q_m$, for $j\ge 1$. For $j
= 0$ it is easy to see that: i) in the case of the prefix automaton,
since $\rho\ge 2$, $q_m\notin D_{j'+1}$ for
$j'\in[\alpha_S(0),\beta_S(0)]$; ii) in the case of the suffix
automaton, if $S[0] = P[m-1]$ then $q_m\in D_{j'+1}$ for
$j'\in[0,\min(\ell_S(0),l_{\rho-1})-1]$, and $q_m\notin D_{j'+1}$
otherwise. Hence, in the case of the suffix automaton, it is enough to
test if $S[0] = P[m-1]$ to know all the matching prefixes in the
interval of the first run of $S$. The idea is then to compute the
configurations $D_j$, restricted to the states with index in
$\mathcal{I}(P)$, corresponding to positions $j\in\mathcal{I}(S)$ only
by reading $S$ run-wise. Observe that it is not possible to detect the
single prefix of $S$ in the language, if any, ending at a position
between $\alpha_S(j-1)$ and $\beta_S(j-1)$ using $D_{\alpha_S(j)}$,
because $q_m\notin D_{\alpha_S(j)}$ if the prefix does not end at
position $\beta_S(j-1)$, or equivalently if $\ell_S(j-1) >
l_{\rho-1}$. To overcome this problem we modify the automata by adding
a self-loop on $q_m$ labeled by $P[m-1]$. In this way, if $q_m$
belongs to a configuration $D$ with index between $\alpha_S(j-1)+1$
and $\beta_S(j-1)+1$ then $q_m$ will be in $D_{\alpha_S(j)}$. Observe
that, if $q_m\in D_{\alpha_S(j)}$, then $q_m\notin D_{\alpha_S(j+1)}$,
since $S[\alpha_S(j)]\neq S[\alpha_S(j + 1)]$.

Let $\bar{D}_j = \{ 1\le i \le \rho\ |\ q_{\alpha_P(i)}\in
D_{\alpha_S(j)}\}\,, $ be the encoding of the configuration of
$\mathcal{A}(P)$ after reading $S[0\,\ldots\, \beta_S(j-1)]$, for
$0\le j\le |\textsc{rle}(S)|$. For example:
\begin{center}
\begin{tabular}{ll}
\hline
$P = cttcct$ & \\
$S = cttccttcct$ & \\
\hline
$\bar{D}_1 = \{ 1 \}$ &
$\bar{D}_2 = \{ 2 \}$ \\
$\bar{D}_3 = \{ 1, 3 \}$ &
$\bar{D}_4 = \{ 2, 4 \}$ \\
$\bar{D}_5 = \{ 1, 3 \}$ &
$\bar{D}_6 = \{ 4 \}$ \\
\hline
\end{tabular}
\end{center}
\begin{center}
\begin{tabular}{|l|lllll|}
\hline
$i$ & 0 & 1 & 2 & 3 & 4 \\
\hline
$\alpha_P(i)$ & 0 & 1 & 3 & 5 & 6 \\
\hline
\end{tabular}
\begin{tabular}{|l|lllllll|}
\hline
$i$ & 0 & 1 & 2 & 3 & 4 & 5 & 6 \\
\hline
$\alpha_S(i)$ & 0 & 1 & 3 & 5 & 7 & 9 & 10 \\
\hline
\end{tabular}
\end{center}
Note that $q_0$ is not represented and that $\bar{D}_0$ is equal to
$\emptyset$ and $\{ 1,\ldots,\rho\}$ for $\mathcal{A}(P)$ and
$\mathcal{S}(P)$, respectively. We now describe how to compute the
configurations $\bar{D}_j$, starting with the automaton
$\mathcal{A}(P)$. The following property easily follows from Equation~\ref{eq:delta}:
\begin{lemma}\label{lemma:trans}
For any $0\le j < |\textsc{rle}(S)|$ and $1\le i \le \rho$, $q_{\alpha_P(i)}\in D_{\alpha_S(j+1)}$ if and only if either
\begin{enumerate}[a)]
\item $q_{\alpha_P(i-1)}\in D_{\alpha_S(j)}\wedge S[\alpha_s(j)] = c_{i-1}\wedge \ell_S(j) = l_{i-1}$, for $i=2,\ldots,\rho-1$;
\item $q_{\alpha_P(i-1)}\in D_{\alpha_S(j)}\wedge S[\alpha_s(j)] = c_{i-1}\wedge \ell_S(j) \ge l_{i-1}$, for $i\in\{1,\rho\}$.
\end{enumerate}
\end{lemma}
Based on the above Lemma and the fact that there is a self-loop on
$q_0$ labeled by $\Sigma$, we have that
\begin{equation}\label{eq:trans}
\begin{array}{ll}
\bar{D}_{j+1} & = \{ i + 1\ |\ i\in\bar{D}_j\cup\{ 0 \} \} \\
& \cap\,\,\{ 1\le i\le \rho\ |\ S[\alpha_S(j)] = c_{i-1} \} \\
& \cap\,\,(\{ 2\le i\le \rho-1\ |\ \ell_S(j) = l_{i-1} \}\cup \{ i\in\{1,\rho\}\ |\ \ell_S(j) \ge l_{i-1} \}) \\
\end{array}
\end{equation}
for $j\ge 0$.
We now show to implement Equation~\ref{eq:trans} efficiently using word-level parallelism.
Let
$$
\begin{array}{ll}
B_1(c) & = \{ 1\le i\le \rho\ |\ c = c_{i-1} \}\,, \\
B_2(l) & = \{ 1\le i\le \rho\ |\ l = l_{i-1} \} \\
& \cup\,\, \{ i\in \{1, \rho\}\ |\ l\ge l_{i-1} \}
\end{array}
$$ for any $c\in\Sigma$ and $1\le l\le |P| + 1$. The set $B_1(c)$
includes the indices of all the runs whose symbol is equal to $c$.
Similarly, The set $B_2(l)$ includes the indices of all the runs whose
length is equal to $l$ and, in addition, the index of the first and/or
last run if the corresponding length is smaller than or equal to $l$.
Note that $B_2(l) = \{1,\rho\}$, for any $l > |P|$; thus, we can
define $B_2$ up to $|P| + 1$ and map any integer greater than $|P|$
onto $|P|+1$. For example, for $P = cttcct$ we have:
$$
\begin{array}{ll}
B_1(c) = \{ 0, 2 \} & B_2(1) = \{ 0, 3 \} \\
B_1(t) = \{ 1, 3 \} & B_2(2) = \{ 0, 1, 2, 3 \} \\
\end{array}
$$
and $B_2(l) = \{ 0, 3 \}$, for $3\le l\le 7$.
We represent the configurations $\bar{D}$ and the sets $B$ as
bit-vectors of $\rho$ bits, denoted with $\bpstyle{D}$ and
$\bpstyle{B}$, respectively. Based on these two definitions,
Equation~\ref{eq:trans} can be rewritten as
$$
\bar{D}_{j+1} = \{ i + 1\ |\ i\in\bar{D}_{j}\cup \{ 0 \} \}\cap B_1(S[\alpha_S(j)])\cap B_2(\min(\ell_S(j), |P| + 1))\,,
$$
which corresponds to the following bitwise operations
$$
\bpstyle{D}_{j+1} = ((\bpstyle{D}_j\ll 1) \OR 0^{\rho-1}1) \AND \bpstyle{B}_1(S[\alpha_S(j)])\AND \bpstyle{B}_2(\min(\ell_S(j), |P| + 1))\,.
$$

We now describe the computation of $\bar{D}_j$ for the automaton
$\mathcal{S}(P)$. The simulation of $\mathcal{S}(P)$ is equivalent to
the one of $\mathcal{A}(P)$, up to state $q_0$ and the first
transition. For $j > 1$, Lemma~\ref{lemma:trans} holds if index $1$ is
handled by case $a$ instead of $b$, which accounts for the fact that
there is no self-loop on $q_0$. We can thus change
Equation~\ref{eq:trans} and the definition of $B_2$ accordingly.
Instead, for $j = 1$ (i.e., the transition on the first run of $S$),
we have
$$
q_{\alpha_P(i)}\in D_{\alpha_S(1)}\iff S[0] = c_{i-1}\wedge (i = \rho\vee \ell_S(0) \le l_{i-1})\,,\text{ for } i=1,\ldots,\rho\,,
$$
and
$$
\bar{D}_1 = \{ 1\le i\le \rho\ |\ S[0] = c_{i-1}\}\cap (\{ 1\le i < \rho\ |\ \ell_S(0)\le l_{i-1}\}\cup\{\rho\})\,,
$$
since, before the first transition, all the states are active because
of the $\varepsilon$-transitions and so state $q_{\alpha_P(i)}$ can be
activated by any state with index between $\alpha_P(i-1)$ and
$\alpha_P(i)-1$ by reading a run with symbol equal to $c_{i-1}$ and
length no larger than $l_{i-1}$, with the exception of $q_m$ which can
be activated with a run of any length because of the self-loop. To
account for this case we define the set $B_3(l) = \{ 1\le i\le
\rho\ |\ l\le l_{i-1}\}\cup \{ \rho \}$, for $1\le l\le |P| + 1$. The
set $B_3(l)$ includes the indices of all the runs whose length is
greater than or equal to $l$ and, in addition, the index of the last
run. Note that $B_3(l) = \{\rho\}$ for $l\ge |P| + 1$, so we can
define $B_3$ up to $|P| + 1$ and map any integer greater than $|P|$
onto $|P| + 1$, as done for $B_2$. Then, we have:
$$
\bar{D}_{1} = B_1(S[0])\cap B_3(\min(\ell_S(0), |P| + 1))\,,
$$
which corresponds to the following bitwise operations
$$
\bpstyle{D}_{1} = \bpstyle{B}_1(S[0])\AND \bpstyle{B}_3(\min(\ell_S(0), |P| + 1))\,.
$$
The computation of a single configuration $\bar{D}_j$ requires
$O(\ceil{\rho / w})$ time. The total time complexity of the simulation
is thus $O(|S| \ceil{\rho / w})$, as the total number of
configurations is $|\textsc{rle}(S)|\le |S|$. The bit-vectors
$\bpstyle{B}$ can be preprocessed in $O(m + (\sigma + m)\ceil{\rho /
  w})$ time and require $O((\sigma + m)\ceil{\rho / w})$ space. The
string $P$ or $S$ can be given in either unencoded or run-length
encoded form. In the former case its run-length encoding
does not need to be stored. It can be computed on the fly in $O(m)$ or
$O(|S|)$ time, using constant space, during the preprocessing or
searching phase.

\begin{figure}[t]
\begin{center}
\begin{scriptsize}
\begin{tabular}{|ll|}
\hline
\multicolumn{2}{|l|}{\textsc{rl-preprocess}$(P)$}\\
\ninstr & $\rho\leftarrow |\textsc{rle}(P)|$ \\
\ninstr & \textbf{for} $c\in\Sigma$ \textbf{do} $\bpstyle{B}_1[c]\leftarrow 0^{\rho}$ \\
\ninstr & \textbf{for} $i\leftarrow 1$ \textbf{to} $|P| + 1$ \textbf{do} $\bpstyle{B}_2[i]\leftarrow 0^{\rho}$\\
\ninstr & $i\leftarrow 0$ \\
\ninstr & \textbf{for} $(c,l)\in \textsc{rle}(P)$ \textbf{do} \\
\ninstr & \qquad $H\leftarrow 0^{\rho-1}1\ll i$ \\
\ninstr & \qquad $\bpstyle{B}_1[c]\leftarrow \bpstyle{B}_1[c]\OR \bpstyle{H}$ \\
\ninstr & \qquad \textbf{if} $i = 0$ \textbf{or} $i = \rho - 1$ \textbf{then} \\
\ninstr & \qquad \qquad $\ell = l$ \\
\ninstr & \qquad \qquad \textbf{for} $j\leftarrow l$ \textbf{to} $|P| + 1$ \textbf{do} \\
\ninstr & \qquad \qquad \qquad $\bpstyle{B}_2[j]\leftarrow \bpstyle{B}_2[j]\OR \bpstyle{H}$ \\
\ninstr & \qquad \textbf{else} $\bpstyle{B}_2[l]\leftarrow \bpstyle{B}_2[l]\OR \bpstyle{H}$ \\
\ninstr & \qquad $i\leftarrow i + 1 $ \\
\ninstr & \textbf{return}$(B_1, B_2, \rho, \ell)$ \\
\hline
\end{tabular}
\setcounter{instr}{0}
\begin{tabular}{cc}
\begin{tabular}{|ll|}
\hline
\multicolumn{2}{|l|}{\textsc{rl-shift-and}$(P, T)$}\\
\ninstr & $(B_1, B_2, \rho, \ell)\leftarrow \textsc{rl-preprocess}(P)$ \\
\ninstr & $\bpstyle{D}\leftarrow 0^{\rho}$ \\
\ninstr & $j\leftarrow 0$ \\
\ninstr & \textbf{for} $(c,l)\in \textsc{rle}(T)$ \textbf{do} \\
\ninstr & \qquad $\bpstyle{D}\leftarrow ((\bpstyle{D}\ll 1)\OR 0^{\rho - 1}1)\AND \bpstyle{B}_1[c]$ \\
\ninstr & \qquad $\bpstyle{D}\leftarrow \bpstyle{D} \AND \bpstyle{B}_2[\min(l,|P|+1)]$ \\
\ninstr & \qquad \textbf{if} $\bpstyle{D}\AND 10^{\rho-1}\neq 0^{\rho}$ \textbf{then} \\
\ninstr & \qquad \qquad \textsf{Output}($j + \ell$) \\
\ninstr & \qquad $j\leftarrow j + l$ \\
& \\
& \\
& \\
& \\
& \\
& \\
& \\
& \\
& \\
\hline
\end{tabular} &
\setcounter{instr}{0}
\begin{tabular}{|ll|}
\hline
\multicolumn{2}{|l|}{\textsc{rl-bndm}$(P, T)$}\\
\ninstr & $(B_1, B_2, \rho, \ell)\leftarrow \textsc{rl-preprocess}(P^r)$ \\
\ninstr & $s\leftarrow m - 1$ \\
\ninstr & \textbf{while} $s < |T|$ \textbf{do} \\
\ninstr & \qquad $\bpstyle{D}\leftarrow 1^{\rho}$ \\
\ninstr & \qquad $b\leftarrow s-m+1$ \\
\ninstr & \qquad \textbf{while} $s + 1 < |T|$ \textbf{and} $T[s] = T[s + 1]$ \textbf{do} \\
\ninstr & \qquad \qquad $s\leftarrow s + 1$ \\
\ninstr & \qquad $j\leftarrow 0$, $k\leftarrow 1$ \\
\ninstr & \qquad \textbf{for} $(c,l)\in \textsc{rle}(T[b\,\ldots\,s]^r)$ \textbf{do} \\
\ninstr & \qquad \qquad $\bpstyle{D}\leftarrow \bpstyle{D} \AND \bpstyle{B}_1[c]$ \\
\ninstr & \qquad \qquad $\bpstyle{D}\leftarrow \bpstyle{D} \AND \bpstyle{B}_2[\min(l,|P|+1)]$ \\
\ninstr & \qquad \qquad \textbf{if} $\bpstyle{D}\AND 10^{\rho-1}\neq 0^{\rho}$ \textbf{then} \\
\ninstr & \qquad \qquad \qquad \textbf{if} ($j + \ell \ge |P|$) \textbf{then} \\
\ninstr & \qquad \qquad \qquad \qquad \textsf{Output}($s - j - \ell$) \\
\ninstr & \qquad \qquad \qquad \textbf{else} $k\leftarrow j + \ell$ \\
\ninstr & \qquad \qquad $\bpstyle{D}\leftarrow \bpstyle{D}\ll 1$ \\
\ninstr & \qquad \qquad $j\leftarrow j + l$ \\
\ninstr & \qquad $s\leftarrow s + m - k$ \\
\hline
\end{tabular} \\
\end{tabular}
\end{scriptsize}
\end{center}
\label{fig:rle-search}
\caption{The variants of \SA and \BNDM based on the run-length encoding.}
\end{figure}

\section{The variants of \SA and \BNDM}

The variants of the \SA and \BNDM algorithms based on the encoding
described in the previous section run in $O(n\ceil{\rho/w})$ and
$O(nm\ceil{\rho/w})$ time, respectively, using $O((\sigma +
m)\ceil{\rho/w})$ space. The encoding of the suffix automaton is
however not ideal in practice, due to the different first transition.
We now describe a variant of \BNDM, based on a modified suffix
automaton, where the first transition of the automaton is equal to the
subsequent ones. Let $j$ be the ending position of a window of length
$m$ in $T$ and let $\bar{j}$ be the minimum position such that
$\bar{j}\ge j$ and $T[\bar{j}]\neq T[\bar{j} + 1]$. In other words,
$\bar{j}$ is the ending position of the run of $\textsc{rle}(T)$
spanning $T[j]$. Consider the window of length $m + \bar{j} - j$
ending at $\bar{j}$. By Lemma~\ref{lemma:single}, if $\rho\ge 2$,
there can be at most one occurrence of $P$ ending at a position
between $j$ and $\bar{j}$. Our idea is to process this larger window
with $\mathcal{S}(P^r)$, finding the single occurrence of $P$ in it,
if any, and the length $k'$ of the longest proper prefix of $P$ ending
at position $\bar{j}$. We then shift the window by $m + \bar{j} - j -
k'$ positions to the right. If $\bar{j} - j < m$, the time needed to
process this window is $O(m)$. Otherwise, if $\bar{j} - j \ge m$,
observe that, since $k' < m$, the symbols of $T$ in the interval
$[j+1,\bar{j}-m]$ are covered by this window only and therefore the time
needed to process this window is $O(m)$, for reading the symbols of
$T$ in the intervals $[j-m+1, j]$ and $[\bar{j}-m+1, \bar{j}]$, plus a
term which, summed over all such windows, is $O(n)$. Hence, the time
complexity of the algorithm remains $O(nm\ceil{\rho/w})$ in the
worst-case.

Suppose that we simulate the automaton $\mathcal{S}(P^r)$ on
$(T[j-m+1\,\ldots\,\bar{j}])^r$ and let $k'$ be the length of the
longest suffix of $T[j-m+1\,\ldots\,\bar{j}]$ which is a proper prefix
of $P$. Observe that, if $k'$ does not correspond to the ending
position of a run of $\textsc{rle}(P)$, i.e., if $P[k'-1] = P[k']$,
then the window corresponding to shift $k'$ does not contain an
occurrence of $P$, because $P[k'-1] = T[\bar{j}]$ and $T[\bar{j}]\neq
T[\bar{j}+1]$ (shifting the window by $m + \bar{j} - j -k'$
corresponds to aligning position $k'$ with $\bar{j}+1$). Indeed, it is
easy to see that the smallest useful shift corresponds to the length
$k'$ of the longest suffix of $T[j-m+1\,\ldots\,\bar{j}]$ which is a
proper prefix of $P$ and such that $P[k'-1]\neq P[k']$.
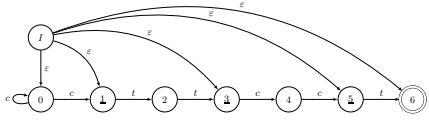
\begin{figure}[t]
\begin{center}
\resizebox{0.47\textwidth}{!}{
\begin{tikzpicture}[node distance=2cm,auto]
\node[state] (q_0) {$0$};
\node[state] (I) [above of=q_0] {$I$};
\node[state] (q_1) [right of=q_0] {$\underline{1}$};
\node[state] (q_2) [right of=q_1] {$2$};
\node[state] (q_3) [right of=q_2] {$\underline{3}$};
\node[state] (q_4) [right of=q_3] {$4$};
\node[state] (q_5) [right of=q_4] {$\underline{5}$};
\node[state,accepting] (q_6) [right of=q_5] {$6$};

\path[->] (q_0) edge node {$c$} (q_1)
                edge [loop left] node {$c$} ()
          (q_1) edge node {$t$} (q_2)
          (q_2) edge node {$t$} (q_3)
          (q_3) edge node {$c$} (q_4)
          (q_4) edge node {$c$} (q_5)
          (q_5) edge node {$t$} (q_6)
          (I) edge node {$\varepsilon$} (q_0)
          (I) edge [bend left] node {$\varepsilon$} (q_1)
          (I) edge [bend left] node {$\varepsilon$} (q_3)
          (I) edge [bend left] node {$\varepsilon$} (q_5)
          (I) edge [bend left] node {$\varepsilon$} (q_6);
\end{tikzpicture}
}
\caption{(a) The automaton $\mathcal{S}_R(P)$ for the pattern $P =
  cttcct$.}
\label{fig:example2}
\end{center}
\end{figure}
Hence, we can modify the automaton $\mathcal{S}(P)$ so that it
recognizes the subset of $\Suff(P)$ $\{ P[i\,\ldots\,m-1]\ |\ i =
0\vee P[i]\neq P[i-1]\}$. To accomplish this, it is enough to remove
all the $\varepsilon$-transitions entering a state $q_i$ with $P[i] =
P[i-1]$. Observe that the automaton can recognize the occurrence of
$P$, if any, ending at a position between $j$ and $\bar{j}$ only if it ends at position $\bar{j}$. To
account for this problem, we add
a self-loop on $q_0$ labeled by $P[0]$. We denote the resulting
automaton with $\mathcal{S}_R(P)$. The transition function of the
$\mathcal{S}_R(P)$ automaton is defined as follows:
    $$
    \delta(q,c) \defAs \begin{cases}
      \{ q_0, q_1 \} & \text{if } q = q_0 \text{ and } c = P[0] \\
    \{q_{i+1}\} & \text{if } q = q_{i} \text{ and } c = P[i] ~~(0 \le i < m)\\
    \{ q_i\ |\ i = 0\vee P[i]\neq P[i-1] \} & \text{if } q = I  \text{ and } c = \varepsilon\\
    \emptyset & \text{otherwise}
    \end{cases}
    $$
Figure~\ref{fig:example2} shows the automaton $\mathcal{S}_R(P)$ for
$P = cttcct$. It is not hard to verify that Lemma~\ref{lemma:trans}
also holds for $\mathcal{S}_R(P)$ and therefore the simulation of this
automaton is analogous to the one of the $\mathcal{A}(P)$ automaton.

The pseudocode of the variants of the \SA and \BNDM algorithms based
on the run-length encoding is shown in Figure~\ref{fig:rle-search}.

\section{Acknowledgments}

We thank Jorma Tarhio for helpful comments.

\bibliographystyle{abbrv}
\bibliography{rle_nfa}

\begin{thebibliography}{10}

\bibitem{DBLP:journals/cacm/Baeza-YatesG92}
R.~A. Baeza{-}Yates and G.~H. Gonnet.
\newblock A new approach to text searching.
\newblock {\em Commun. {ACM}}, 35(10):74--82, 1992.

\bibitem{DBLP:journals/algorithmica/Baeza-YatesN99}
R.~A. Baeza{-}Yates and G.~Navarro.
\newblock Faster approximate string matching.
\newblock {\em Algorithmica}, 23(2):127--158, 1999.

\bibitem{DBLP:journals/iandc/CantoneFG12}
D.~Cantone, S.~Faro, and E.~Giaquinta.
\newblock A compact representation of nondeterministic (suffix) automata for
  the bit-parallel approach.
\newblock {\em Inf. Comput.}, 213:3--12, 2012.

\bibitem{DBLP:books/ox/CrochemoreR94}
M.~Crochemore and W.~Rytter.
\newblock {\em Text Algorithms}.
\newblock Oxford University Press, 1994.

\bibitem{DBLP:conf/wea/DurianPST10}
B.~Durian, H.~Peltola, L.~Salmela, and J.~Tarhio.
\newblock Bit-parallel search algorithms for long patterns.
\newblock In {\em Experimental Algorithms, 9th International Symposium, {SEA}
  2010, Ischia Island, Naples, Italy, May 20-22, 2010. Proceedings}, pages
  129--140, 2010.

\bibitem{DBLP:journals/jda/FredrikssonG09}
K.~Fredriksson and S.~Grabowski.
\newblock Average-optimal string matching.
\newblock {\em J. Discrete Algorithms}, 7(4):579--594, 2009.

\bibitem{DBLP:journals/ipl/Hyyro08}
H.~Hyyr{\"{o}}.
\newblock Improving the bit-parallel {NFA} of baeza-yates and navarro for
  approximate string matching.
\newblock {\em Inf. Process. Lett.}, 108(5):313--319, 2008.

\bibitem{DBLP:journals/algorithmica/HyyroN05}
H.~Hyyr{\"{o}} and G.~Navarro.
\newblock Bit-parallel witnesses and their applications to approximate string
  matching.
\newblock {\em Algorithmica}, 41(3):203--231, 2005.

\bibitem{DBLP:journals/siamcomp/KnuthMP77}
D.~E. Knuth, J.~H.~M. Jr., and V.~R. Pratt.
\newblock Fast pattern matching in strings.
\newblock {\em {SIAM} J. Comput.}, 6(2):323--350, 1977.

\bibitem{DBLP:journals/jea/NavarroR00}
G.~Navarro and M.~Raffinot.
\newblock Fast and flexible string matching by combining bit-parallelism and
  suffix automata.
\newblock {\em {ACM} Journal of Experimental Algorithmics}, 5:4, 2000.

\bibitem{DBLP:journals/cacm/WuM92}
S.~Wu and U.~Manber.
\newblock Fast text searching allowing errors.
\newblock {\em Commun. {ACM}}, 35(10):83--91, 1992.

\end{thebibliography}

\end{document}